\documentclass[review]{elsarticle}

\usepackage{lineno,hyperref}
\modulolinenumbers[1]

\journal{Information Sciences}

\usepackage{graphicx}
\usepackage{amsthm}
\theoremstyle{definition}
\newtheorem{propr}{Property}
\usepackage{amsmath}
\usepackage{amssymb}
\usepackage{multirow}
\usepackage{xcolor}
\newtheorem{thm}{Teorema}[section]

\newtheorem{proposition}[thm]{Proposition}
\newtheorem{definition}[thm]{Definition}

\DeclareMathOperator{\argmax}{argmax}

\begin{document}

\begin{frontmatter}

\title{A new approach in model selection for ordinal target variables}

\author{Elena Ballante}
\address{Department of Mathematics, University of Pavia}
\address{Via Ferrata 5, 27100 Pavia, Italy}
\author{Pierpaolo Uberti}
\address{Department of Economics, University of Genova}
\address{Via Vivaldi 5, 16126 Genova, Italy}
\author{Silvia Figini\corref{correspondingauthor}}
\cortext[correspondingauthor]{Corresponding author}
\ead{silvia.figini@unipv.it}
\address{Department of Political and Social Sciences, University of Pavia}
\address{Corso Strada Nuova 65, 27100 Pavia, Italy}



%
%
%

\begin{abstract}
This paper introduces a novel approach to assess model performance for predictive models characterized by an ordinal target variable in order to satisfy the lack of suitable tools in this framework. Our methodological proposal is a new index for model assessment which satisfies mathematical properties and can be easily computed. \\
In order to show how our performance indicator works, empirical evidence achieved on a toy examples and simulated data are provided. On the basis of results at hand, we underline that our approach discriminates better for model selection with respect to performance indexes proposed in the literature. 

\end{abstract}

\begin{keyword}
Ordinal classification \sep Performance metric \sep Model selection
\end{keyword}

\end{frontmatter}


\section{Introduction}
\label{intro}
Evaluation measures are widely used in predictive models to compare different algorithms, thus providing the selection of the best model for the data at hand. \\
Performance indicators can be used to assess the performance of a model in terms of accuracy, discriminatory power and stability of the results. The choice of indicators to made model selection is a fundamental point and many approaches have been proposed over the years (see e.g. \cite{Adams,Bradley,Hand2009}).\\
Restricting to binary target variables, distinct criteria for comparing the performance of classification models are available (see \cite{Hand1997, Hand2000, review, AccuracyFROC}).\\
Multi-class classification models are generally evaluated averaging binary classification indicators (see \cite{AUCmulticlass, review, perfMeas}) and in the literature there is not a clear distinction among them with respect to multi-class nominal and ordinal targets (e.g. \cite{simpleapproach,Gaudette,Pang}).\\
While in the model definition stage for ordinal target variable there are different approaches in the literature (see \cite{agresti,ordinal1,ordinal2,ordinal3}), for the model selection there is a lack of adequate tools (\cite{performance}). \\
In our opinion, performance indicators should take into account the nature of the target variable, especially when the dependent variable is ordinal. This leads us to propose a new class of measures to select the best model in predictive contexts characterized by a multi-class ordinal target variable, using the misclassification errors coupled with a measure of uncertainty on the prediction. \\
The paper is structured as follow: Section \ref{literature} reviews the metrics most used in literature; Section \ref{definition} shows our methodological proposal and proves some mathematical properties; Section \ref{toy} explains how our proposal works in two toy examples; Section \ref{simulation} reports the empirical evidence obtained on simulated data. Conclusions and further ideas for research are summarized in Section \ref{conclusions}. 

\section{Review of the literature for ordinal dependent variable}
\label{literature}
The most popular measures of performances in ordinal predictive classification models are based on AUC (Area Under the ROC curve), accuracy (expressed in terms of correct classification) and MSE (Mean Square Error) (see \cite{Gaudette} and \cite{contr} among others). 
The accuracy (percentage of correct predictions over total instances) is the most used evaluation metric for binary and multi-class classification problems (\cite{AccuracyFROC}), assuming that the costs of the different misclassifications are equal.\\
The AUC for multi-class classification is defined in \cite{AUCmulticlass} as a generalization of the AUC (based on the probabilistic definition of AUC); it suffers of different weaknesses also in the binary classification problem (\cite{ROCint}) and it is cost-independent, assumption that can be viewed as a weakness when the target is ordinal.\\
The mean square error (MSE) measures the difference between prediction values and observed values in regression problems using an Euclidean distance. 
MSE can be used in ordinal predictive models, converting the classes of the ordinal target variable $y$ in integers and computing the difference between them and it does not takes into account the ordering in a predictive model characterized by ordinal classes in the response variable.\\
Furthermore, it is well known that in imbalanced data characterized by under-fitting or over-fitting the mean square error could provide trivial results (see \cite{review}). \\

\section{A new index for model performances evaluation and comparison for ordinal target}
\label{definition}
Let $\mathbf{y}=\{y_1,..,y_N\}$ be a test set for the ordinal target variable $Y$, where $y_i\in \{1,...,M\}$ (with $M$ number of classes ordered of the target variable) and let $\mathbb{X}$ be the $N \times p$  data matrix, where $N$ is the number of observations and $p$ the number of covariates.\\
The output of a predictive model is a matrix $P=\{p_{ij}\}$, where $0\leq p_{ij} \leq 1$, which contains the probability that observation $i$ belong to the class $j$, estimated by the model under evaluation.\\
Standard multi-class classification rules assign the observation $i$ to the class $j=\argmax_l\{p_{i,l}\}$. \\ 
In order to introduce our proposal, the definitions of classification function and error interval are required.	

\begin{definition}[Classification function]
	\label{class_func} 
	Let observations $\{1,...,N\}$ grouped by the estimated classes $\hat{y_i}=j$. For each class, sort the observations in a non-increasing order with respect to $p_{i,j}$. The vector of indexes $i$ of the observations is a permutation of the original vector, according to the ordering defined above.
	For a given model, the classification function is a piecewise constant function $f_{mod}:[0,1]\to \{1,..,M\}$ such that $f_{mod}([\frac{i-1}{N},\frac{i}{N}))=y_i$ for $i \in \{1,...,N\}$.	\\
	
\end{definition}

As a special case, the \textit{perfect classification function}, is a piecewise constant function $f_{exact}:[0,1]\to \{1,..,M\}$ such that each estimated class corresponds to the real class identified by $\mathbf{y}$.\\
Note that the function $f_{exact}$ is unique except for permutation of the observations in the same estimated class.\\

The error interval in each class can be derived as the interval between the first misclassified observation and the end of the observations in that estimated class. 
\begin{definition}[Error Interval]
	\label{error}
	Suppose that the range corresponding to the estimated class $j$ is $[n_{j-1},n_j)$, let $\tilde{i_j} \in \{n_{j-1},...,n_j\}$ the first misclassified observation. So the error interval is defined as $[\frac{\tilde{i_j}}{N},\frac{n_j}{N})$ and its length is $e_j=\frac{n_j-\tilde{i_j}}{N}$.\\
	If no misclassification occurs in $[n_{j-1},n_j)$, the error interval is defined as an empty set and the length is $e_j=0$.
\end{definition}

Consider, for example, $N=10$ observations and a three levels target variable ($M=3$). Suppose that a predictive model returns the predictions as in Table \ref{table:output}. For each observation, the real class is reported.

\begin{table}[!ht]
	\begin{center}
		\begin{tabular}{c|c|c|c|c|c}
			Observation	& \multicolumn{3}{c|}{Probabilities} & Estimated Class & Real Class\\
			& Class 1 & Class 2 & Class 2 &   &\\
			\hline
			1 & 0.288 & 0.174 & \textbf{0.538} & 3 & 1\\
			2 & 0.325 & \textbf{0.478} & 0.197 & 2 & 2\\
			3 & \textbf{0.828} & 0.013 & 0.159 & 1 & 1\\
			4 & 0.310 & 0.106 & \textbf{0.584} & 3 & 3\\
			5 & 0.120 & 0.262 & \textbf{0.618} & 3 & 3\\
			6 & \textbf{0.426} & 0.167 & 0.407 & 1 & 3\\
			7 & \textbf{0.849} & 0.126 & 0.025 & 1 & 2\\
			8 & \textbf{0.520} & 0.401 & 0.079 & 1 & 1\\
			9 & 0.147 & \textbf{0.670} & 0.183 & 2 & 2\\
			10 & 0.142 & \textbf{0.593} & 0.265 & 2 & 3\\  
			
		\end{tabular}
	\end{center}
	\caption{Example}
	\label{table:output}
\end{table}

The classification function is derived grouping the observations in the estimated class as: \{3,6,7,8\} in Class 1, \{2,9,10\} in Class 2 and \{1,4,5\} in Class 3.
In each group the observations are sorted with respect to the probability of the estimated class. For the group 1 the probabilities are 0.828, 0.426, 0.849, 0.520 respectively, then the ordered group is: \{7,3,8,6\}. Following the same rule the group 2 becomes \{9,10,2\} and group 3 \{5,4,1\}.\\
The final sequence of observations can be written as in Table \ref{table:ordered}.

\begin{table}[!ht]
	\begin{center}
		\begin{tabular}{|c|c|c|c|c|c|c|c|c|c|c|}
			\hline
			i & 7 & 3 & 8 & 6 & 9 & 10 & 2 & 5 & 4 & 1\\
			\hline
			$\tilde{i}$ & 1 & 2 & 3 & 4 & 5 & 6 & 7 & 8 & 9 & 10 \\
			\hline
			x & 0 & 0.1 & 0.2 & 0.3 & 0.4 & 0.5 & 0.6 & 0.7 & 0.8 & 0.9 \\
			\hline
			y & 2 & 1 & 1 & 3 & 2 & 3 & 2 & 3 & 3 & 1 \\
			\hline
			$\hat{y}$ & 1 & 1 & 1 & 1 & 2 & 2 & 2 & 3 & 3 & 3 \\	
			\hline
		\end{tabular}
	\end{center}
	\caption{Index construction}
	\label{table:ordered}
\end{table}

The classification function and the corresponding perfect classification function are depicted in Figure \ref{fig:function} and Figure \ref{fig:perffunction} respectively.\\
\begin{figure}[!htb]
	\begin{minipage}{0.48\textwidth}
		\includegraphics[width=0.9\textwidth]{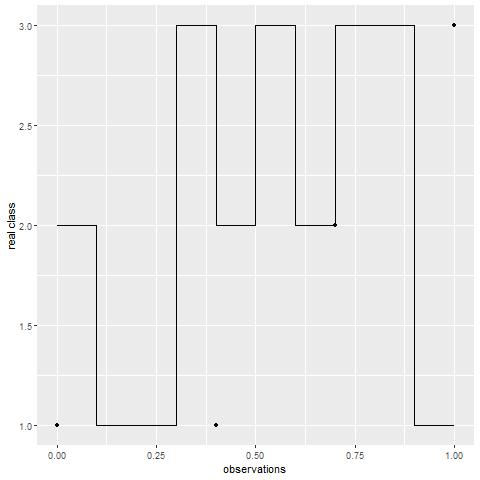}
		\caption{Classification function}\label{fig:function}
	\end{minipage}\hfill
	\begin{minipage}{0.48\textwidth}
		\includegraphics[width=0.9\textwidth]{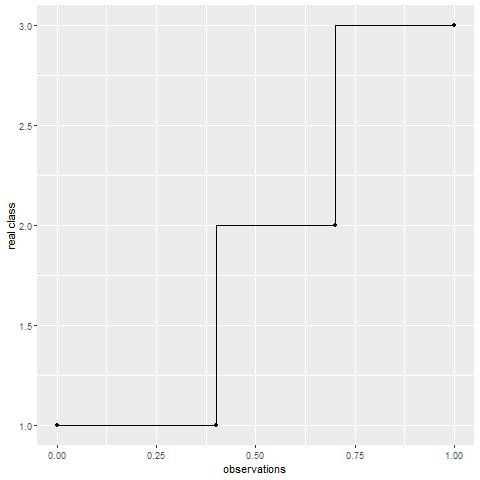}
		\caption{Perfect classification function}\label{fig:perffunction}
	\end{minipage}
\end{figure}

In order to define the three error intervals, as a preliminary step we identify the intervals of observations related to each estimated class: $[0,0.4)$ for Class 1, $[0.4,0.7)$ for Class 2,  $[0.7,1)$ for Class 3. From Table \ref{table:ordered}, in the estimated Class 1 the first error corresponds to the first observation, so the error interval is $[0,0.4)$, in the estimated Class 2 the first error corresponds to the observation 6, then the error interval is $[0.5,0.7)$ and in the estimated Class 3 the first error corresponds to the observation 10 and the error interval is $[0.9,1)$.  \\

Starting from Definition \ref{class_func} and Definition \ref{error}, Definition \ref{index} introduces a new index for model performance evaluation in predictive models characterized by an ordinal target variable.

\begin{definition}[Index]
	\label{index}
	\begin{equation*}
	I=\sum_{j=1}^{M} w_j\int_{\frac{n_{j-1}}{N}}^{\frac{n_j}{N}} |(f_{mod}(x)-f_{exact}(x))|dx
	\end{equation*}
	where  $l_j=n_j-n_{j-1}$ is the length of the $j^{th}$ class in the domain, $w_j=\frac{e_j}{l_j}$ and $0\leq w_j \leq 1$.
\end{definition}
On the basis of the previous example, we can compute the value for the index introduced in Definition \ref{index}: the three integral results are (0.3, 0.1, 0.2)  and the corresponding weights are (1, 0.67, 0.33), thus $ I=0.433$.\\
The index satisfies the following properties.
\label{prop}
\begin{propr}
	$I \in [0,+\infty)$. \\
	$I=0$ if and only if $f_{mod}=f_{exact}.$
\end{propr}

\begin{proof}
	\begin{equation*}
	I=\sum_{j=0}^{M-1} w_j\int_{\frac{n_{j-1}}{N}}^{\frac{n_j}{N}} |(f_{mod}-f_{exact})(x)|dx \geq \sum_{j=0}^{M-1}\frac{n_j - \tilde{i}_j}{N} |f_{mod}-f_{exact}| \frac{n_j-n_{j-1}}{N}
	\end{equation*}
	and
	\begin{itemize}
		\item	$n_j \geq \tilde{i}_j$, 
		\item	$n_j>n_{j-1}$
	\end{itemize} 
	by definition, than we can conclude that $I \geq 0$.\\
	We prove also that $I=0$ if and only if $f_{mod}=f_{exact}$.\\	
	
	$I=0$ $\implies$ $w_j=0$ or $\int_{\frac{n_{j-1}}{N}}^{\frac{n_j}{N}} |(f_{mod}-f_{exact})(x)|dx =0 \quad \forall$ $j$ in $\{1,...,M-1\}$.\\
	\begin{itemize}
		\item $w_j=0  \Longleftrightarrow \tilde{i}_j=n_j$, i.e there are not classification errors, so $f_{mod}=f_{exact}$ in class $j$.\\
		\item $\int_{\frac{n_{j-1}}{N}}^{\frac{n_j}{N}} |(f_{mod}-f_{exact})(x)|dx =0 \Longleftrightarrow f_{mod}=f_{exact}$ in the class $j$.
	\end{itemize}
	
	So we can conclude that $I=0 \implies f_{mod}=f_{exact}$.\\
	The other implication is trivial.
\end{proof}

\begin{propr}
	$I$ has a sharp upper bound $M-1$\\
	The upper bound $M-1$ is reached if and only if $M=2$ (binary classification).
\end{propr}
\begin{proof}
	\begin{equation*}
	\begin{split}
	I &=\sum_{j=0}^{M-1} w_j\int_{\frac{n_{j-1}}{N}}^{\frac{n_j}{N}} |(f_{mod}-f_{exact})(x)|dx \leq \sum_{j=0}^{M-1} 1 \cdot \int_{\frac{n_{j-1}}{N}}^{\frac{n_j}{N}} |(f_{mod}-f_{exact})(x)|dx \leq \\ & \leq \max_{x}{|(f_{mod}-f_{exact})(x)|} \sum_{j=0}^{M-1}\frac{n_j-n_{j-1}}{N} \leq M-1
	\end{split}
	\end{equation*}
	If $M=2$ we obtain $|(f_{mod}-f_{exact})(x)|=1$ $\forall x \in [0,1]$ so that $I=M-1$.
	If $M>2$, $|(f_{mod}-f_{exact})(x)|>1$ for at least one class (by construction) the inequality is strict.
\end{proof}

\begin{proposition}\label{prop_max}
	$I \leq K$,\\
	where $K$ is defined as
	\begin{equation*}
	K=\sum_{i=1}^{M} l_i \max\{M-i,i-1\}
	\end{equation*}
\end{proposition} 
\begin{proof} 
	The maximum value is reached when the worst classification is obtained, i.e. when all observations are associated to the fairest class.
	If this happens, the error interval is long as the class domain, so $w_j=1 \, \forall j=1,...,M$ and each integral is the sum is a rectangle with basis the class domain $l_j$ and height the maximum height reachable.
	
\end{proof}

\begin{definition}[Normalized index]
	\begin{equation*}
	I_n=\frac{1}{K}\sum_{j=0}^{M-1} w_j\int_{\frac{n_{j-1}}{N}}^{\frac{n_j}{N}} |(f_{mod}-f_{exact})(x)|dx
	\end{equation*}
	where $K$ is the maximum defined in the Proposition \ref{prop_max}.\\
	So $0 \leq I_n \leq 1$.
\end{definition}

In the previous example, $K=1.7$ and the corresponding value of the defined normalized index is $0.255$.

\begin{proposition}
	The accuracy is a special case of the index introduced in Definition \ref{index}.	
\end{proposition}
\begin{proof}
	The accuracy is $acc=p_{err}=\frac{\#\{\text{misclassified observations}\}}{N}$ i.e. the proportion of misclassified observations.\\
	Setting $M=2$, from the Proposition \ref{prop_max}, $K=1$.\\ 
	$max_x |f_{mod}(x)-f_{exact}(x)|=1$, each error weights $\frac{1}{N}$ if $w_1=w_2=1$ and $I_n=p_{err}$.
\end{proof}

\begin{propr}[Monotonicity] \label{prop_monot}
	Consider a classification $C$ with $\epsilon$ misclassification and $N$ observations.
	Operating a transformation of the classification $C$ in $C'$ where an observation right classified is changed in a misclassification, the index $I_n$ becomes higher.
\end{propr}
\begin{proof}
	In the classification $C'$, $\epsilon'$=$\epsilon+1$ are misclassified observations: the $\epsilon$ observations misclassified in $C$ plus a new misclassification. Suppose that the new misclassification is the observation $i$ that is classified in the class $j'$ instead of the real class $j$.\\
	All the components in the sum of the index $I_n$ remain unchanged except for the $j^{th}$, thus obtaining $I_n^j$.
	So
	\begin{equation*}
	I_n^j=w_j\int_{\frac{n_{j-1}}{N}}^{\frac{n_j}{N}} |f_{mod}(x)-f_{exact}(x)| dx
	\end{equation*}
	Looking at each of the two elements in the product:
	\begin{itemize}
		\item $w_j'\geq w_j$\\
		Two different cases are possible: if the probability associated to the $i^{th}$ observations is less or equal than the probability of the first error, the error interval $w_j'= w_j$; on the other hand, the error interval become larger, thus $w_j'> w_j$.
		\item $|f'_{mod}-f_{exact}| > |f_{mod}-f_{exact}|$\\
		In $C'$ there is one misclassification more than in $C$, so the distance between $f_{mod}$ and $f_{exact}$ increases. 
	\end{itemize}
	We can conclude that $I_n^{'j} \geq I_n^j$.
\end{proof}
We remark that in the Proposition \ref{prop_monot} the vice versa does not hold, i.e. if $I_{mod1}\geq I_{mod2}$ we can not make conclusion on the number of misclassified observations in the two classifications.\\

\section{Toy examples}
\label{toy}
In order to show how our index works with respect to the indexes proposed in the literature toy examples are reported in this section with the main aim of discussing the behaviour in terms of model selection of our index with respect to AUC, accuracy and MSE.\\ 
$Y$ is a target variable characterized by $M=3$ levels $y_i \in \{1,2,3\}$ and model 1 and model 2 are two competitive models under comparison. 

\subsection{First toy example}
In the first toy example we take into account the ordinal structure of the target variable $Y$.
Table \ref{table:Toy1a} and Table \ref{table:Toy1b} are the corresponding confusion matrices for model 1 and model 2.
It is clear that the model 2 makes a better classification than model 1.

\begin{table}[h]
	\begin{center}
		\begin{tabular}{c|c|ccc|}
			& & \multicolumn{3}{c|}{Actual} \\
			\hline
			&	& 1 & 2 & 3  \\
			\hline
			\parbox[t]{2mm}{\multirow{3}{*}{\rotatebox[origin=c]{90}{Predict}}}
			& 1 & 5 & 0 & 1 \\
			&2 & 0 &  7 & 0\\
			&	3 & 0 & 0  &7 \\
		\end{tabular}
	\end{center} 
	\caption{Confusion matrix model 1}
	\label{table:Toy1a}
\end{table}

\begin{table}[h]
	\begin{center}
		\begin{tabular}{c|c|ccc|}
			& & \multicolumn{3}{c|}{Actual} \\
			\hline
			&	& 1 & 2 & 3  \\
			\hline
			\parbox[t]{2mm}{\multirow{3}{*}{\rotatebox[origin=c]{90}{Predict}}} 
			& 1 & 5 & 1 & 0 \\
			& 2 & 0 &  6 & 0\\
			& 3 & 0 & 0  &8 \\
		\end{tabular}
	\end{center} 
	\caption{Confusion matrix model 2} 
	\label{table:Toy1b}
\end{table}

\begin{table}[h]
	\begin{center}
		\begin{tabular}{|c|c|c|c|c|c|}
			\hline
			Model & Proposed Index & Normalized Index & AUC & accuracy & MSE\\
			\hline
			1 & 0.083 & 0.051 & 0.956 & 0.950 & 0.200\\
			\hline
			2 & 0.042 & 0.025 &  0.956 & 0.950 & 0.050\\
			\hline
		\end{tabular}
	\end{center}
	\caption{Results} 
	\label{table:ResultsToy1}
\end{table}
For the sake of comparison, for each model the AUC, the accuracy, the MSE and our index are computed as summarized in Table \ref{table:ResultsToy1}.\\
We remark that looking at Table \ref{table:ResultsToy1} the values obtained for the AUC and the accuracy indexes for model 1 and model 2 are exactly equal, thus, in terms of model choice, model 1 and model 2 are indifferently. Our index highlights a difference in terms of performance between the two models under comparison and it selects model 2 as the best one.

\subsection{Second toy example}
The second toy example considers the probability assigned to each observation. 
In practical applications where we need also to evaluate how much uncertainty is associated to a prediction, the starting point considers the probability that the new observation belongs to the estimated class.\\
From Table \ref{table:ConfMatr2}, Model 1 and model 2 assign an observation of the first class to the second one. The first classification assigns a higher probability to the misclassified observation than the second. Then we can conclude that model 2 is better than model 1 for data at hands.\\

\begin{table}[h]
	\begin{center}
		\begin{tabular}{c|c|ccc|}
			& & \multicolumn{3}{c|}{Actual} \\
			\hline
			&	& 1 & 2 & 3  \\
			\hline
			\parbox[t]{2mm}{\multirow{3}{*}{\rotatebox[origin=c]{90}{Predict}}} & 1 & 5 & 0 & 0 \\
			& 2 & 0 &  7 & 0\\
			& 3 & 1 & 0  & 7 \\
			
		\end{tabular}
	\end{center} 
	\caption{Confusion matrix}
	\label{table:ConfMatr2}
\end{table}

From Table \ref{table:ResultsToy2} both models are equivalent in terms of MSE and accuracy, thus on the basis of classical measures model 1 and model 2 are indifferent. Our index reports different values for the models under comparison and select model 2 as the best one.
\begin{table}[h]
	\begin{center}
		\begin{tabular}{|c|c|c|c|c|c|}
			\hline
			Model & Proposed Index & Normalized Index & AUC & accuracy & MSE \\
			\hline
			1 & 0.083 & 0.051 & 0.956 & 0.950 & 0.200 \\
			\hline
			2 & 0.017 & 0.010 &  0.983 & 0.950 & 0.200 \\
			\hline
		\end{tabular}
	\end{center}
	\caption{Results} 
	\label{table:ResultsToy2}
\end{table}

\section{Empirical evaluation on simulated data}
\label{simulation}
In order to show how our proposal works in model selection, this section reports the empirical results achieved on a simulated dataset.\\ 
The simulated dataset is composed of three covariates obtained by a Monte Carlo simulation and an ordinal target variable with $M=5$, as reported in Table \ref{tab:1}. The sample size is $N=7500$.
\begin{table}[!ht]
	\centering
	\begin{tabular} {|c|ccccc|}
		\hline
		y & 1 & 2 & 3 & 4 & 5 \\
		\hline
		x1 & N(2,1.5) & N(3,1)  & N(4,1.5) & N(5,1)  & N(6,1)  \\
		\hline
		x2  & N(1,2.5)  & N(5,2)  & N(7,2.5)  & N(8.5,2)  & N(9.5,2) \\
		\hline
		x3  & &&U(0,3)&&\\	
		\hline
	\end{tabular}
	\caption{Simulated data structure.}
	\label{tab:1}
\end{table}

Five different models are under comparison:
\begin{itemize}
	\item Ordinal logistic regression (Ord Log),
	\item Classification tree (Tree),
	\item Support vector machine (SVM),
	\item Random forest (RFor),
	\item k- Nearest Neighbour (kNN).
\end{itemize}

For each model AUC, accuracy, MSE and our index are computed.

Table \ref{tab:2} reports, in terms of out of sample, the values of the metrics under comparison obtained for each model using a 10-fold cross validation.\\ 

\begin{table}[!ht]
	\centering
	\begin{tabular}{|c|c|c|c|c|c|}
		\hline
		Model & Proposed Index & Normalized index & AUC & Accuracy & MSE\\
		\hline
		Ord Log & 0.450 & 0.141 & 0.864	& 0.577 & 0.571\\
		\hline
		Tree & 0.487 & 0.146 & 0.835 & 0.585 & 0.654\\
		\hline
		SVM & 0.439 & 0.135 & 0.871 & 0.589 & 0.564\\
		\hline
		RFor & 0.493 & 0.151 & 0.855 & 0.569 & 0.672 \\
		\hline
		kNN & 0.003 & 0.001 & 0.999 & 0.977 & 0.024\\
		\hline
	\end{tabular}
	\caption{Model comparison}
	\label{tab:2}
	
\end{table}

For sake of clarity, Table \ref{tab:3} shows the resulting ranks for the models, using the results obtained for the four metrics under comparison.\\
\begin{table}[!ht]
	\centering
	\begin{tabular}{|c|c|c|c|c|}
		\hline
		Model & Proposed Index/Normalized & AUC & Accuracy & MSE\\
		\hline
		Ord Log & 3	& 3	& 4 & 3\\
		\hline
		Tree & 4 & 5 & 3 & 4\\
		\hline
		SVM & 2 & 2	& 2 & 2\\
		\hline
		RFor & 5 & 4 & 5 & 5 \\
		\hline
		kNN & 1 & 1 & 1 & 1\\
		\hline
	\end{tabular}
	\caption{Results in terms of ranking.}
	\label{tab:3}
\end{table}
We can see that the k-nearest neighbour is classified as the best model according to all the indexes employed for model choice. Furthermore, from table \ref{tab:2} the k-nearest neighbour outperforms the other models. The Support vector machine is considered the second-best model with respect to all performance indicators. 
The rest of the models under comparison are ranked differently with respect to the evaluation metrics adopted. 

\section{Conclusions}
\label{conclusions}
A new performance indicator is proposed to compare predictive classification models characterized by ordinal target variable.\\
Our index is based on a definition of a classification function and an error interval. A normalized version of the index is derived. The empirical evidence at hands underlined that our index discriminates better among different models with respect to classical measures available in the literature. \\
Our index can be used coupled with other metrics for model performance for model selection.\\
From a computational point of view a further idea of research will consider the implementation of our index in a new R package. In terms of application we think that our index could be directly incorporate in the process of assessment for predictive analytics.

\end{document}